\documentclass[letterpaper, 11pt]{amsart}

\usepackage{amsthm}
\usepackage{amsmath}
\usepackage{amsfonts}
\usepackage{amssymb}
\usepackage{bbm}
\usepackage{hyperref} 
\usepackage{cleveref}

\newtheorem{proposition}{Proposition}

\newtheorem{corollary}{Corollary}
\newtheorem{lemma}{Lemma}
\theoremstyle{definition}
\newtheorem{problem}{Problem}
\newtheorem{definition}{Definition}
\newtheorem{example}{Example}
\newtheorem{remark}{Remark}

\newcommand{\range}[1]{\{1,\dots,#1\}}
\newcommand{\id}{\mathbbm{1}}
\newcommand{\card}[1]{|{#1}|}
\newcommand{\Q}{\mathbb{Q}}
\newcommand{\F}{{\mathbb{F}_q}}
\newcommand{\e}{{\{e\}}}

\DeclareMathOperator{\Hom}{Hom}
\DeclareMathOperator{\Aut}{Aut}
\DeclareMathOperator{\Ker}{Ker}
\DeclareMathOperator{\Img}{Im}
\DeclareMathOperator{\GL}{GL}
\DeclareMathOperator{\M}{M}
\DeclareMathOperator{\wt}{wt}
\DeclareMathOperator{\swc}{swc}

\title[Extension property for vector space alphabets]{When the extension property does not hold for vector space alphabets}
\author{Serhii Dyshko}
\email{dyshko@univ-tln.fr}

\begin{document}
\maketitle

\section{Introduction}
In our previous work \cite{d-when-the-ep-does-not-hold} we showed that the extension property holds for linear codes over a finite module alphabet with respect to the symmetrized weight composition if and only if the socle of the alphabet is cyclic. For alphabets with a noncyclic socle we proved the existence of an unextendable isometry. However, our construction was implicit.

This paper deals with a particular case of an alphabet with a noncyclic socle, a vector space alphabet. Here we give an explicit construction of a code with an unextendable isometry, see \Cref{sec:main-construction}. Moreover, the code and the isometry have nice additional properties, see \Cref{thm:properties-of-vs-code} and \Cref{thm-unextnedability}. Unlike the general case, where we do not know the length of the constructed code, for vector space alphabet we show that the length can be relatively small. For the symmetrized weight composition built on the trivial group we give the lover bound on the length of a code with an unextendable isometry, see \Cref{thm-bound-for-trivial-group}.


In \cite{d-when-the-ep-does-not-hold} we also introduced the notion of $G$-pseudo-injectivity for finite modules. It is related to the extension property for codes of the length $1$. We characterize this property in the context of vector spaces, see \Cref{thm:pi-ep}.

\section{Preliminaries}
Let $\F$ be a finite field and let $A$ be an $\F$-linear vector space of dimension $\ell$, where $q$ is a prime power and $\ell$ is a positive integer. Denote by $\GL(A)$ the group of all $\F$-linear automorphism of $A$. This group is isomorphic to the group $\GL_\ell(\F)$ of $\ell \times \ell$ invertible matrices with entries from $\F$. By introducing a basis in $A$ we identify $\GL(A)$ and $\GL_\ell(\F)$.

Let $G$ be a subgroup of $\GL(A)$. The group $G$ acts on the vector space $A$ (from the right). Denote by $A/G$ the set of orbits of the action of $G$ on $A$.

A \emph{closure} $\bar{G}$ of a subgroup $G < \GL(A)$ is the set of those elements in $\GL(A)$ that preserve the orbits of $G$,
\begin{equation*}
	\bar{G} = \{ g \in \GL(A) \mid \forall O \in A/G, g(O) = O  \}\;.
\end{equation*}
We say that $G$ is \emph{closed} if $G = \bar{G}$.	
More on a closure of a group and its properties see in \cite{wood-aut-iso}.

Let $n$ be a positive integer and consider a vector space $A^n \cong \mathbb{F}_q^{n\ell}$.
Define a map $\swc_G: A^n \times A/G \rightarrow \Q$, called the \emph{symmetrized weight composition}, built on the group $G$. For each $a \in A^n$, $O \in A/G$,
\begin{equation*}
\swc_G(a)(O) =  \card{ \{i \in \range{n} \mid a_i \in O\} }\;.
\end{equation*}
The \emph{Hamming weight} $\wt: A^n \rightarrow \{0, \dots, n\}$ is a function that counts the number of nonzero coordinates.
There is always a zero orbit $\{0\}$ in $A/G$. For each $a \in A^n$, $\swc_G(a)(\{0\}) = n - \wt(a)$.

Consider an $\F$-linear code $C \subseteq A^n$ and an $\F$-linear map $f: C\rightarrow A^n$. The map $f$ is called an \emph{$\swc_G$-isometry} if $f$ preserves $\swc_G$. We call $f$ a \emph{Hamming isometry} if $f$ preserves the Hamming weight.

Let $W$ be a vector space isomorphic to $C$.
Let $\lambda \in \Hom_\F(W, {A^n})$ be a map such that $\lambda(W) = C$.
Present the map $\lambda$ in the form $\lambda = (\lambda_1,\dots, \lambda_n)$, where $\lambda_i \in \Hom_\F(W,{A})$ is a projection on the $i$th coordinate, for $i \in \range{n}$.
Define $\mu = f\lambda \in \Hom_\F(W,A^n)$.

A map $h:A \rightarrow A$ is called \emph{$G$-monomial} if there exist a permutation $\pi \in S_n$ and automorphisms $g_1, g_2, \dots, g_n \in \bar{G}$ such that for any $a \in A^n$
\begin{equation*}
h\left( (a_1, a_2, \dots,a_n) \right) = \left( g_1(a_{\pi(1)}), g_2(a_{\pi(2)}) \dots, g_n(a_{\pi(n)})\right)
\end{equation*}

It is not difficult to show that a map $f \in \Hom_\F(A^n, A^n)$ is an $\swc_G$-isometry if and only it is a $G$-monomial map.

We say that $A$ has an \emph{extension property} with respect to $\swc_G$ if for any code $C \subseteq A^n$, each $\swc_G$-isometry $f \in \Hom_\F(C,A^n)$ extends to a $G$-monomial map. 

It is a well-known fact that vector spaces satisfy the property of \emph{pseudo-injectivity}: any isomorphism between two subspaces in $A$ extends to an automorphism of the whole space $A$.
A generalization of this property follows.
\begin{definition}\label{def:psinj}
	A vector space $A$ is $G$-\emph{pseudo-injective} if for any subspace $V \subseteq A$, any $\F$-linear injective map  $f: V \rightarrow A$ that preses the orbits of $G$, i.e. for each $O \in A/G$,
	$f(V \cap O) \subseteq O$,
	extends to an element of $\bar{G}$.
\end{definition}

In our previous work \cite{d-when-the-ep-does-not-hold} we proved the following two propositions.
\begin{proposition}\label{thm:isometry-criterium-vs}
	The map $f \in \Hom_\F (C,A^n)$ is an $\swc_G$-isometry if and only if, for any $O \in A/G$, the following equality holds,
	\begin{equation}\label{eq-main-space-equation}
	\sum_{i=1}^n \id_{\lambda_i^{-1}(O)} = \sum_{i=1}^n \id_{\mu_i^{-1}(O)}\;.
	\end{equation}	
	If $f$ extends to a $G-$monomial map, then there exists a permutation $\pi \in S_n$ such that for each $O \in A/G$ the equality holds,
	\begin{equation}\label{eq-condition}
	\lambda_i^{-1}(O) = \mu_{\pi(i)}^{-1}(O)\;.
	\end{equation}
	If $A$ is $G$-pseudo-injective and \cref{eq-condition} holds, then $f$ extends to a $G$-monomial map.
\end{proposition}

\begin{proposition}\label{thm:pi-ep}
	A vector space $A$ is $G$-pseudo-injective if and only if for any code $C \subset A^1$, each $\swc_G$-isometry $f \in \Hom_\F(C,A)$ extends to a $G$-monomial map.
\end{proposition}
We prove here several additional properties of vector spaces.
\begin{proposition}
	Any vector space $A$ is $\GL(A)$-pseudo-injective and $\e$-pseudo-injective, where $\e$ is a trivial subgroup.
\end{proposition}
\begin{proof}
	First, note that these subgroups are closed.
	The set $A/\GL(A)$ has only two orbits, $\{0\}$ and $A \setminus \{ 0 \}$. As it was stated above, an isomorphism between two different subspaces extends to an automorphism of the ambient space $A$.
	
	Since $A/\e = \{ \{ x \} \mid x \in A \}$, for any subspace $V$, any injective map $f: V\rightarrow A$, such that $f(x) = x$, for all $x \in V$, extends to a trivial map $e \in G$.
\end{proof}

\begin{lemma}\label{lemma:a2}
	Let $U \subseteq V$ be a subspace of dimension smaller than $2$ and let $G$ be a subgroup of $\GL(V)$. Any injective map $f: U \rightarrow V$ that preserves the orbits of $G$ extends to an element of $G$.
\end{lemma}
\begin{proof}
	If $U = \{ 0 \}$, then any $f: \{0\} \rightarrow V$ is a trivial map and hence extends to an element of $\bar{G}$.
	Let $u \in V \setminus \{0\}$ be such that $U = \langle u \rangle_\F$. Let $O$ be an orbit in $V/G$ such that $u \in O$. Since $f$ preserves the orbits of $G$, there exists $v \in O$ such that $v = f(u)$ and there exists $g \in G$ such that $v = g(u)$. Then, for any $\lambda \in \F$, $f(\lambda u) = \lambda f(u) = \lambda v = \lambda g(u) = g (\lambda u)$. But $\{ \lambda u \mid \lambda \in \F \} = U$ and therefore, $f$ extends to an element of $G$
\end{proof}

\begin{remark}
	In the same way as in the proof of \Cref{lemma:a2} we can prove the following generalized fact. Consider a ring $R$ with identity and a finite $R$-module $A$. It is true that for any cyclic submodule $U \subseteq A$ any injective map $f \in \Hom_R(U,A)$ which preserves the orbits of $G \leq \Aut_R(A)$, extends to a $G$-monomial map.
\end{remark}

\section{Extension theorem for vector space alphabets}\label{sec:vec-space-ep}
\subsection{Main construction}\label{sec:main-construction}
Let $A$ and $W$ be two $\F$-linear vector spaces of dimension $2$ and $4$ correspondingly. Fix bases in $A$ and $W$. Consider the trivial subgroup $\e < \GL(A)$. Each orbit of $\e$ contains only one point and the group $\e$ is closed.

Let $\chi(x) = x^2 + \alpha x + \beta$ be an irreducible polynomial over $\F$, where $\beta, \alpha \in \F$. Then $Q: \F \times \F \rightarrow \F$, $(a,b) \mapsto a^2 + \alpha ab + \beta b^2$, is a non-degenerate quadratic form.
	
Let $\mathbb{P}_1(\F)$ be a projective line, $\card{\mathbb{P}_1(\F)} = q + 1$.
For any $[a:b] \in \mathbb{P}_1(\F)$, define a matrix $\Lambda_{[a:b]} \in \M_{4 \times 2}(\F)$,
\begin{align*}
		\Lambda_{[a:b]} =
		\frac{1}{Q(a,b)}\left(\begin{array}{cc}
			-ba & -b^2\\
			a^2 & ab\\
			\beta b^2 & -\alpha b^2 - ab\\
			-\beta a b & a^2 + \alpha ab
		\end{array}\right)\;,
\end{align*}
and let $M_{[a:b]} \in \M_{4 \times 2}(\F)$ be a matrix obtained from $	\Lambda_{[a:b]}$ by flipping the second and the third row,
\begin{align*}
M_{[a:b]} =
\frac{1}{Q(a,b)}\left(\begin{array}{cc}
-ba & -b^2\\
\beta b^2 & -\alpha b^2 - ab\\
a^2 & ab\\
-\beta a b & a^2 + \alpha ab
\end{array}\right)\;.
\end{align*}
The definitions are correct, i.e. the matrices do not depend on the choice of class representatives.
	
For any $p \in \mathbb{P}_1(\F)$, define $\F$-linear maps $\lambda_{p}, \mu_{p}: W \rightarrow A$ as $\lambda_{p}(w) = w \Lambda_{p}$ and $\mu_{p}(w) = w M_{p}$, for all vectors $w \in W$.
Let $\lambda, \mu \in \Hom_\F(W,A^n)$ be defined as $\lambda = (\lambda_p)_{p \in \mathbb{P}_1(\F)}$ and $\mu = (\mu_p)_{p \in \mathbb{P}_1(\F)}$. Note that the map $\lambda$ is injective.
Define a $\F$-linear code $C \subset A^n$ as the image $C = \Img \lambda$. Define a $\F$-linear map $f: C \rightarrow A^n$ as $f = \mu\lambda^{-1}$. The generator matrix of $C$ can be presented as a concatenation of the $\Lambda$-matrices and the generator matrix of the image $f(C)$ is a concatenation of the $M$-matrices.

Consider the irreducible polynomial $\bar{\chi}(x) = \chi(-x) = x^2 - \alpha x + \beta$ and
consider a finite field extension $\mathbb{F}_{q^2}$ of $\F$, $\mathbb{F}_{q^2} = \F[x] / (\bar{\chi}(x))$. Identify $A$ and $\mathbb{F}_{q^2}$ by introducing the basis $1, \omega$ in the last, where $\bar{\chi}(\omega) = 0$.
\begin{proposition}\label{thm:properties-of-vs-code}
The code $C$ is an $\mathbb{F}_{q^2}$-linear $[q+1,2]_{\mathbb{F}_{q^2}}$ MDS code.
The map $f$ is an $\F$-linear $\swc_\e$-automorphism of $C$ that does not extend to a $\GL_{2}(\F)$-monomial map.
\end{proposition}
\begin{proof}
Prove that the map $f$ is an $\swc_\e$-isometry.
For any $(x,y) \in A$, ${[a:b]} \in \mathbb{P}_1(\F)$ denote
\begin{align*}
V_{[a:b]} &= \langle (a, b, 0, 0), (0,0,a,b) \rangle_\F\;,\\
U_{[a:b]} &= \langle (a, 0, b, 0), (0,a,0,b) \rangle_\F\;,
\end{align*} and check that,
\begin{align*}
\lambda_{[a:b]}^{-1}(\{(x,y)\}) &= (-\alpha x - \beta y, x, x, y) + V_{[a:b]}\;,\\
\mu_{[a:b]}^{-1}(\{(x,y)\}) &= (-\alpha x - \beta y, x, x, y) + U_{[a:b]}\;.
\end{align*}
Indeed, for $\lambda_{[a:b]}$ calculate,
\begin{align*}
&(a,b,0,0) \Lambda_{[a:b]} = \frac{1}{Q(a,b)} (-ba^2 + ba^2, -a b^2 + a b^2) = (0,0)\;,\\
&(0,0,a,b) \Lambda_{[a:b]} = \frac{1}{Q(a,b)} (\beta a b^2- \beta a^2 b, -\alpha ab^2 - ba^2 + ba^2 + \alpha a b^2) = (0,0)\;,\\
&(-\alpha ,1,1,0) \Lambda_{[a:b]} = \frac{1}{Q(a,b)} (\alpha ba + a^2 + \beta b^2,
\alpha b^2  + ab - \alpha b^2 - a b) = (1,0)\;,\\
&(-\beta,0,0,1) \Lambda_{[a:b]} = \frac{1}{Q(a,b)} (\beta b a - \beta a b,
\beta b^2 + a^2 + \alpha a b) = (0,1)\;.
\end{align*}
In the same way make calculations for the maps $\mu_{[a:b]}$.

In \cite{d-geom} we proved that $\sum_{p \in \mathbb{P}_1(\F)} \id_{V_{p}} = \sum_{p \in \mathbb{P}_1(\F)} \id_{U_{p}}$. Also, for any $(x,y) \in A$, there is the equality in \cref{eq-main-space-equation},
\begin{equation*}
\sum_{[a:b] \in \mathbb{P}_1(\F)} \id_{(-\alpha x - \beta y, x, x, y) +  V_{[a:b]}} = \sum_{[a:b] \in \mathbb{P}_1(\F)} \id_{(-\alpha x - \beta y, x, x, y) + U_{[a:b]}}\;.
\end{equation*}
Hence, by \Cref{thm:isometry-criterium-vs}, $f$ is an $\swc_{\e}$-isometry. However, $f$ does not extends even to a $\GL(A)$-monomial map since condition (\ref{eq-condition}) does not hold for $0 \in A$. 

Prove that $f$ is an automorphism of $C$. Denote $\vec{v}_i = \lambda(\vec{e}_i)$, $i \in \{1,2,3,4\}$, where $\vec{e}_i \in W$ is the vector with 1 on the $i$th position and zeros elsewhere. The map $f$ fixes $\vec{v}_1$ and $\vec{v}_4$ and flips $\vec{v}_2$ and $\vec{v}_3$, hence $f(C) = C$.

The code $C$ is a $\F$-linear code in $\mathbb{F}_{q^2}^n$. Note that $\vec{v}_3 = \omega \vec{v}_1$ and $\vec{v}_4 = \omega \vec{v}_2$. Indeed, for any ${[a:b]} \in \mathbb{P}_1(\F)$,
\begin{align*}
(-ba - b^2 \omega)\omega &= -ba \omega - b^2 (\alpha \omega - \beta) = \beta b^2 - (\alpha b^2 + ab) \omega\;,\\
(a^2 + ab \omega)\omega &= a^2 \omega +ab (\alpha \omega - \beta) = -\beta ab - (a^2 + \alpha ab) \omega\;.
\end{align*}
So, $C$ is an $\mathbb{F}_{q^2}$-linear code. However, $f: C \rightarrow C$ is not an $\mathbb{F}_{q^2}$-linear map.

Considering $C$ as a code in a Hamming space $\mathbb{F}_{q^2}^n$, it is equivalent, with a $\GL_2(\F)$-monomial equivalence, to a code observed in \cite{d3}, where we proved that $C$ is a $[q + 1, 2]_{\mathbb{F}_{q^2}}$ MDS code (for MDS codes see \cite[p.~319]{macwilliams}). 
\end{proof}

\subsection{Related properties}\label{sec:additional-properties}
Let us observe properties of subcodes of $C$. It is clear that any $\F$-linear $\swc_\e$-isometry of any one-dimensional code in $A^n$ is extendable. But, in general, it is not true for codes of dimension greater than 1. Note that the restrictions of $f$ on subcodes $\langle \vec{v}_1,\vec{v}_2, \vec{v}_3 \rangle_\F$ and $\langle \vec{v}_2, \vec{v}_3, \vec{v}_4 \rangle_\F$ are unextendable $\swc_\e$-automorphisms.

Consider a subcode $C' = \langle \vec{v}_1,\vec{v}_2 \rangle_\F \subset C$, and the restriction of $f$ on $C'$. The code $C'$ is constant-weight (with respect to the Hamming weight), for any $x \in C'\setminus\{0\}$, $\wt(x) = n - 1$ and it is the smallest $\F$-linear code (both with respect to the size and the length) for which there exists an unextendable $\swc_{\e}$-isometry, see \Cref{thm-bound-for-trivial-group} later in this paper.

For the code $C$ we see that 
\cref{eq-main-space-equation} calculated at any point $(x,y) \in A$ is just a shifted in $W$, by a vector $(-\alpha x - \beta y, x , x, y)$, \cref{eq-main-space-equation} calculated at $(0,0)$. 
For $C'$ this is not the case. Let $W' \subset W$ be a subspace generated by $\langle (1,0,0,0), (0,1,0,0) \rangle_\F$ and consider the $\F$-linear maps $\lambda, \mu : W' \rightarrow A^n$. Clearly, $C' = \lambda(W')$. Denote a line $L_{[a:b]} = \{ (x,y) \in W' \cong \mathbb{F}^2 \mid {[x:y]} = {[a:b]} \}$, for ${[a:b]} \in \mathbb{P}_1(\F)$, and calculate the preimages in $W'$,

\begin{equation*}
\lambda_{[a:b]}^{-1}((x,y)) = 
\left\lbrace\begin{array}{ll}
L_{[a:b]},& \text{ if } (x,y) = (0,0);\\
\emptyset,&\text{ if } [x:y] \neq [a:b];\\
(-\alpha x - \beta y, x) + L_{[a:b]},&\text{ if } [x:y] = [a:b].
\end{array} \right.
\end{equation*}
\begin{equation*}
\mu_{[a:b]}^{-1}((x,y)) =	\left\lbrace\begin{array}{ll}
(-\alpha x - \beta y - \frac{a}{b} x, x - \frac{a}{b}y), & \text{if } b \neq 0;\\
W', & \text{if } b = 0 \text{ and } (x,y) = 0;\\
\emptyset, & \text{if } b = 0 \text{ and } (x,y) \neq 0.
\end{array} \right.
\end{equation*}
Note that if ${[x:y]} = {[a:b]}$, then
\begin{align*}
\lambda_{[a:b]}^{-1}((x,y)) 
= \{ \mu_{[a:b]}^{-1} ((x,y)) \mid {[a:b]} \neq {[1:0]} \}\;.
\end{align*}
Therefore, for any nonzero $(x,y) \in A$, \cref{eq-condition} calculated for the orbit $\{(x,y)\}$ becomes,
\begin{equation*}
\sum_{{[1:0]}\neq [a:b] \in \mathbb{P}_1(\F)} \id_{\mu^{-1}_{[a:b]}(\{(x,y)\})} + \id_{\emptyset} = \id_{\lambda^{-1}_{[x:y]}(\{(x,y)\})} + q\id_{\emptyset}\;,
\end{equation*}
which is an equality, as a covering of a shifted line by points. For the orbit $\{(0,0)\}$ we have,
\begin{equation*}
q\id_{ (0,0) } + \id_{W'} = \sum_{[a:b] \in \mathbb{P}_1(\F)} \id_{L_{[a:b]}}\;,
\end{equation*}
which is an equality, as a covering of a two dimensional plane by lines. The $\swc_\e$-isometry $f$ is not extendable to a $\GL(A)$-monomial map because condition (\ref{eq-condition}) does not hold.

To illustrate the constructions we give an example for a finite field $\mathbb{F}_2$.
\begin{example}\label{ex:2}
	Let $q = 2$. Consider a quadratic non-degenerate form $Q: \mathbb{F}_2 \times \mathbb{F}_2 \rightarrow \mathbb{F}_2$, $Q(a,b) = a^2 + ab + b^2$, where $a,b \in \mathbb{F}_2$. Consider a projective line $\mathbb{P}_1(\mathbb{F}_2) = \{ {[0:1]}, {[1:0]}, {[1:1]} \}$. Write down the matrix $\Lambda_{[a:b]}$, for each $[a:b] \in \mathbb{P}_1(\mathbb{F}_2)$,
\begin{align*}
\Lambda_{[0:1]} =
\left(\begin{array}{cc}
0 & 1\\
0 & 0\\
1 & 1\\
0 & 0
\end{array}\right)\;,\;
\Lambda_{[1:1]} =
\left(\begin{array}{cc}
1 & 1\\
1 & 1\\
1 & 0\\
1 & 0
\end{array}\right)\;,\;
\Lambda_{[1:0]} =
\left(\begin{array}{cc}
0 & 0\\
1 & 0\\
0 & 0\\
0 & 1
\end{array}\right)\;.
\end{align*}
Define an $\mathbb{F}_2$-linear code $C \subseteq (\mathbb{F}_2^2)^3$ as an $\mathbb{F}_2$-span of rows in the left matrix,
\begin{equation*}
\left(\begin{array}{ccc}
0 1 & 1 1 & 0 0\\
0 0 & 1 1 & 1 0\\
1 1 & 1 0 & 0 0\\
0 0 & 1 0 & 0 1
\end{array}\right) \xrightarrow{f} 
\left(\begin{array}{ccc}
0 1 & 1 1 & 0 0\\
1 1 & 1 0 & 0 0\\
0 0 & 1 1 & 1 0\\
0 0 & 1 0 & 0 1
\end{array}\right)\;,
\end{equation*}
and define an $\mathbb{F}_2$-linear map $f: C \rightarrow (\mathbb{F}_2^2)^3$ in the same way as described in \Cref{sec:main-construction}.
Easy to check that $f$ is an unextendable $\swc_\e$-isometry.
\end{example}

\subsection{General result for vector space alphabets}
We proved in \cite{d1} several facts about $\F$-linear codes over a vector space alphabet $A \cong \mathbb{F}_q^\ell$. More precisely, combining the ideas of Wood and Dinh, we showed that a code, for which the extension property fails to hold, with respect to the Hamming weight, should have its length to be greater than $q$. 
\begin{proposition}\label{thm-unextnedability}
Let $A$ be an $\F$-linear vector space of dimension $\ell$ greater than 1.
Let $G$ be a subgroup of $\GL(A)$.
For any $n > q$ there exist a code $C \subset A^n$ and an $\F$-linear $\swc_G$-automorphism $f: C \rightarrow C$ that does not extend to a $\GL(A)$-monomial map.
\end{proposition}
\begin{proof}
	Let $B$ be a $\F$-linear vector space of dimension 2. 
	Consider a code $C \subset B^{q + 1}$ and an $\swc_{\e}$-automorphism $f: C \rightarrow C$ constructed in \Cref{sec:main-construction}. There exists an embedding $B \subseteq A$, so we can consider the code $C$ as a code in $A^{q + 1}$ and map $f$ as an $\swc_{\e}$-isometry $C \rightarrow C$, unextendable to a $\GL(B)$-monomial map.
	
	The map $f$ does not extends even to a $\GL(A)$-monomial map. Really, consider the restriction of $f$ on the subcode $C'$, observed in \Cref{sec:additional-properties}. Since the codes $C$ and $C'$ have different number of zero-columns, $f: C' \rightarrow A^{q+1}$ does not extend to a $\GL(A)$-monomial map and so does the lifted map $f: C \rightarrow A^{q+1}$.
	
	Adding $n - q - 1$ zero columns to $C$ and redefining $f$ as zero on these additional coordinates will not change the unextendability of the $\swc_\e$-isometry. Since $\e \leq G$, the map $f$ is an $\swc_G$-isometry.
\end{proof}

\subsection{Trivial group case}
For $\swc_\e$ we can improve \Cref{thm-unextnedability}.
\begin{lemma}\label{lemma:intersection-of-shifted}
	Let $U, V$ be two vector spaces in $W$ and let $x,y \in W$. Then, $\card{ x + U \cap y + V }$ equals or $0$ either $\card{U \cap V}$.
\end{lemma}
\begin{proof}
	If $x + U \cap y + V = \emptyset$, then we are done. Assume that there exists $w \in x + U \cap y + V$. Therefore there exist $u \in U$ and $v \in V$ such that $w = x + u = y + v$. Shift the sets $x + U$ and $y + U$ by the vector $-w$. Note that the size of the intersection of the shifted sets remains unchanged and calculate $\card{ -w + x + U \cap -w + y + V} = \card{U \cap V}$.
\end{proof}

\begin{proposition}\label{thm-bound-for-trivial-group}
	Let $A$ be an $\F$-linear vector space of dimension $\ell$ greater than 1. Let $n \leq q$ and let $C \subseteq A^n$ be an $\F$-linear code.
	Each $\F$-linear $\swc_\e$-isometry $f: C \rightarrow A^n$ extends to a $\e$-monomial map.
\end{proposition}
\begin{proof}
	Let $f$ be an unextendable $\F$-linear $\swc_{\e}$-isometry of $C$. If $f$ is unextendable as a Hamming isometry, then, in \cite{d1} we proved that $n > q$. So in the proof we can assume that $f$ is an extendable Hamming isometry. In \cite{d1} we proved that in this case there exists a permutation $\pi \in S_n$ such that $\Ker \lambda_i = \Ker \mu_{\pi(i)}$, for $i \in \range{n}$. Making a corresponding reindexing, we can assume that the permutation $\pi$ is trivial.
	
	Use the second part of the statement in \Cref{thm:isometry-criterium-vs}, which implies the existence of an orbit $O = \{x\} \neq \{0\}$, such that \cref{eq-condition} does not hold but \cref{eq-main-space-equation} holds.
	
	For a $K$-linear map $\sigma:W \rightarrow A$, let us describe the structure of preimages $\sigma^{-1}(O)$. It is true that or $\sigma^{-1}(x) = w + \Ker \sigma$, for some point $w \in W$, such that $\sigma(w)=x$, either $\sigma^{-1}(x) = \emptyset$, if $x \not\in \Img \sigma$.
	
	Since \cref{eq-condition} does not hold, without loss of generality, suppose that $\mu^{-1}_1(x)$ is covered nontrivially by sets $\lambda^{-1}_i(x)$, $i \in I$, for some $I \subseteq \range{n}$.
	In other words, $\mu^{-1}_1(x) = \bigcup_{i \in I} \lambda_i^{-1}(x) \cap \mu^{-1}_1(x)$ and $\emptyset \subset \lambda_i^{-1}(x) \cap \mu^{-1}_1(x) \subset \mu_1^{-1}(x)$. 
	
	Using \Cref{lemma:intersection-of-shifted}, for any $i \in I$, $\card{\lambda_i^{-1}(x) \cap \mu^{-1}_1(x)} \leq \card{\Ker \lambda_i \cap \Ker \mu_1}$. 
	Calculate the size of $\Ker \mu_1$, using the fact that $\Ker \lambda_i \cap \Ker \mu_1$, $i \in I$, and $\Ker \mu_1$ can intersect each other at most by a hyperplane of $\Ker \mu_1$,
	\begin{equation*}
	\card{\Ker\mu_1} = \card{\mu^{-1}_1(x)} \leq \sum_{i \in I} \card{\lambda^{-1}_i(x) \cap \mu^{-1}_1(x)} \leq \card{I} \frac{\card{\Ker\mu_1}}{q}\;,
	\end{equation*}
	and hence $\card{I} \geq q$. From our assumption, $\Ker \mu_1 = \Ker \lambda_1$ and thus $1 \not\in I$. Therefore, $n = \card{\range{n}} \geq \card{I \cup \{1\}} > q$. 
\end{proof}

\section{Pseudo-injectivity of vector spaces}
In this section we characterize the property of $G$-pseudo-injectivity for vector spaces in order to construct unextendable $\swc_G$-isometries of codes of length $1$.
\subsection{Poset of partitions into orbits}
\begin{definition}
A partition $P_1$ of the set $S$ is said to be \emph{finer} than a partition $P_2$ of the same set $S$, denoted $P_1 \leq P_2$, if any set in $P_2$ is a disjoint union of sets from $P_1$.
\end{definition}
The set of all possible partitions of the set $A$ with the relation $\leq$ is a poset.

\begin{example}
	Let $S = \{1, 2, 3\}$. Then for three partitions $P_1 = (1,2,3)$, $P_2 = (1,2)(3)$ and $P_3 = (1)(2)(3)$ the following holds,
	$$ P_3 \leq P_2 \leq P_1 \;.$$
\end{example}

Let $S$ be a set and let $G$ be a finite group acting on $S$ (from the right). The set of orbits of the action, denoted as $S/G$, induces a partition $P_G$ of the set $S$ into a union of disjoint subsets. Let $H$ be a subgroup of $G$. The subgroup $H$ inherits an action on $S$ from $G$. Obviously, $P_H \leq P_G$.

Denote by $\mathcal{P}$ the poset $(\{P_H \mid H \leq G\}, \leq)$.
Define the union of the partitions $P_{H_1}$ and $P_{H_2}$, denoted $P_{H_1} \cup P_{H_2}$, as the smallest element in $\mathcal{P}$ that contains both $P_{H_1}$ and $P_{H_2}$. Such an element exists since $P_G$ is the greatest element in $\mathcal{P}$.

\begin{proposition}\label{thm:poset-closed}
	Let $H = \langle h_1, \dots, h_n \rangle \leq G$. Then $P_H = \bigcup_{i=1}^n P_{\langle h_i \rangle}$.
\end{proposition}
\begin{proof}
	Since $\langle h_i \rangle \leq H$, it is true that $P_{\langle h_i \rangle} \leq P_H$, for all $i \in \range{n}$, and thus $P = \bigcup_{i=1}^n P_{\langle h_i \rangle} \leq P_H$.
	
	Let $J$ be a subgroup of $G$ that contains all elements in $G$ that fix all orbits in $P$. Obviously, $J$ is closed and $P = P_J$. It is easy to see that $h_i \in J$, for all $i \in \range{n}$, that means $H \leq J$. Thereby, $P_H \leq P_J = P$.
\end{proof}

\begin{proposition}\label{thm:closed-subgroup}
	Let $H \leq G$ be a closed subgroup and let $J = \{ g \in H \mid g(x) = x \}$ for some $x \in S$. The subgroup $J$ is closed and $P_J \leq P_H$.
\end{proposition}
\begin{proof}
	The group $J$ is a subgroup of $H$, which means $P_J \leq P_H$. Let $g \in G$ be an element such that $g$ fixes all orbits in $P_J$. Then $g \in H$ and $g(x) = x$, so $g \in J$. Therefore $J = \bar{J}$. 
\end{proof}

\subsection{Spaces of dimension that differs from 3}
Let $V$ be an $\F$-linear vector space.
\begin{proposition}\label{thm:lt2}
	If $\dim_\F V \leq 2$, then for any subgroup $G \leq \GL(V)$ the space $V$ is $G$-pseudo-injective.
\end{proposition}
\begin{proof}
	Let $U$ be a subspace of $V$ and let $f : U \rightarrow V$ be a map that preserves the orbits of $G$.
	If $U = V$, then $f$ is an element of $G$. If $U$ is a proper subspace, $\dim _\F U < 2$ and from \Cref{lemma:a2}, $f$ extends to an element of $G \leq \bar{G}$. Hence $V$ is $G$-pseudo-injective for any $G$.
\end{proof}

Let $m \geq 2$ be an integer and let $n = m^2$. Consider the following block-diagonal $n \times n$ matrix,
\begin{equation*}
T = \left(\begin{matrix}
M & 0 & \dots & 0\\
0 & M & \dots & 0\\
\vdots & \vdots & \ddots & \vdots\\
0 & 0 & \dots & M
\end{matrix}\right)\;,
\end{equation*}
consisting of $m$ blocks on the diagonal, where $M \in \GL_m(\F)$.
The matrix $T$ generates a multiplicative subgroup $\langle T \rangle < \GL_n(\F)$.

\begin{lemma}\label{lemma:a1}
	The group $\langle T \rangle$ is closed.
\end{lemma}
\begin{proof}
	Let $B \in \GL_n(\F)$ be an element that preserves the orbits of $\langle T \rangle$, i.e., $B \in \bar{\langle T \rangle}$. This means that for all vectors $v = (x_1, \dots, x_m) \in V$, where $x_i$ is an $m$-dimensional vector, there exists an integer $p_v$ such that $$vB = vT^{p_v}\;.$$ Let $B$ has the block form, $B = (B_{ij})_{i,j \in \range{m}}$, where each $B_{ij} \in \M_{m \times m}(\F)$. 
	Put $v = v_i = (0, \dots, 0, x, 0 ,\dots, 0)$ with the vector $x \in \mathbb{F}_q^m$ on $i$th position, $i \in \range{m}$. Then the equality becomes,
	\begin{equation*}
	(x B_{i1}, x B_{i2}, \dots, x B_{im}) = (0, \dots, 0, xM^{p_{v_i}}, 0, \dots, 0)\;.
	\end{equation*}
	Since the vector $x \in \mathbb{F}_q^m$ is arbitrary, $B_{ij} = 0$ for all $i \neq j \in \range{m}$.
	
	Put $v = v_{i,j} \in \mathbb{F}_q^n$, the vector with $x \in \mathbb{F}_q^m$ on $i$th and $j$th positions, $i \neq j \in \range{m}$, and $0$ everywhere else.
	We get then two equalities,
	$$xB_{ii} = x M^{p_{v_{i,j}}} \;, xB_{jj} = x M^{p_{v_{i,j}}}\;,$$
	for each $x \in \mathbb{F}_q^m$. Therefore, $B_{ii} = B_{jj}$ for all $i,j \in \range{m}$.
	
	Put $v = v_0 = ( e_1, e_2, \dots, e_m)$, where $e_i \in \mathbb{F}_q^m$ is a vector with $1$ on the $i$th position and $0$ everywhere else.
	Then, for any $i \in \range{m}$,
	$$e_i B_{11} = e_i M^{p_{v_{0}}}\;,$$
	which implies $B_{11} = M^{p_{v_0}}$, for some integer $p_{v_0}$. Hence $B = T^{p_{v_0}}$, $B \in \langle T \rangle$ and thus $\bar{\langle T \rangle} \subseteq \langle T \rangle$. Therefore, $\langle T \rangle$ is closed.
\end{proof}

\begin{remark}\label{rem:1}
	Let $m \geq 2$ be an integer and let $n > m^2$ be another integer. Consider the block diagonal $n \times n$ matrix,
	\begin{equation*}
	T' = \left(\begin{matrix}
	T & 0\\
	0 & I
	\end{matrix}\right)\;,
	\end{equation*}
	where $T$ is an $m^2 \times m^2$ matrix defined above and $I \in \GL_{n - m^2}(\F)$ is an identity matrix.
	Based on the proof of \Cref{lemma:a1}, it is easy to see that the group $\langle T' \rangle$ is a closed subgroup of $\GL_n(\F)$.
\end{remark}

\begin{proposition}\label{thm:geq4}
	If $\dim_\F V \geq 4$, then $V$ is not $\langle T' \rangle$-pseudo-injective.
\end{proposition}
\begin{proof}
	Let $M \in \GL_2(\F)$ be a matrix of multiplicative order $q^2 - 1$. Such a matrix exists: choose $M$ to be a multiplication matrix of an element $\omega \in \mathbb{F}_{q^2}$, where $\mathbb{F}_{q^2}$ is a finite field extension of $\F$ and $\omega$ is a generator element of the multiplicative group $\mathbb{F}_{q^2}^*$.
	
	Denote $G = \langle T' \rangle$. Choose a subspace $U \subset V$ of dimension $2$ generated by the vectors $a = (1,0,0,\dots,0)$ and $b = (0,1,0,\dots,0)$. The subspace $U$ intersects only 2 orbits of $V/G$, $\{0\}$ and $U \setminus \{0\}$. Really, since we chose $M$ to be a multiplication by a generator of a multiplicative group of $\mathbb{F}_{q^2}$, any nonzero vector in $U$ can by mapped by some element in $G$ to any nonzero vector in $U$.
	
	A map $f: U \rightarrow V$  defined by, $f(a) = b$ and $f(b) = a$, is an automorphism of the vector space $U \subset V$. Obviously, $f$ preserves the orbits of $G$.
	
	By contradiction, assume that $f$ extends to an element of $\bar{G}$. As we proved in \Cref{lemma:a1} (see also \Cref{rem:1}), $G = \bar{G}$. Hence, there exists an element $g \in G$ such that for any $u \in U$, $f(u) = g(u)$. From the construction, the elements of $G$ acts by a multiplication on the first pair of coordinates. Since we are interested in only two first coordinates, let $a \in \mathbb{F}_{q^2}$ be an element that corresponds to the vector $(1,0)$ and let $b \in \mathbb{F}_{q^2}$ corresponds to $(0,1)$. Hence, there exists an element $x \in \mathbb{F}_{q^2}$ such that $f(a) = xa$ and $f(b) = xb$. But from the other side, $f(a) = b$ and $f(b) = a$. Combining these equalities we get,
	\begin{equation*}
	x a = b \text{ and } xb = a.
	\end{equation*}
	From this, $x^2 = 1$ and $a^2 = b^2$. Rewriting the last equality, $a^2 - b^2 = (a - b) (a + b) =0 $ that implies $a = b$ or $a = - b$ which is impossible, since $a$ and $b$ are linearly independent over $\F$. From the contradiction, $f$ does not extend to an element of $\bar{G}$ and hence $V$ is not $G$-pseudo-injective. 
\end{proof}
\subsection{Three-dimensional spaces}
At first observe the case $q \neq 2$ and consider the following matrix $X \in \GL_3(\F)$,
\begin{equation*}
X = \left(\begin{matrix}
M & 0\\
0 & \det M  
\end{matrix}\right)\;,
\end{equation*}
where $M$ is from the proof of \Cref{thm:geq4}, and is a multiplication matrix of the element $\omega \in \mathbb{F}_{q^2}$. Additionally assume that $M$ is represented in the basis $1, \omega$.

\begin{proposition}\label{thm:dim3}
	Let $V$ be a $3$-dimensional $\F$-linear vector space, $q \neq 2$. The space $V$ is not $\langle X \rangle$-pseudo-injective.
\end{proposition}
\begin{proof}
	Denote $G = \langle X \rangle < \GL_3(\F)$. Since $\det: \GL_2(\F) \rightarrow \mathbb{F}_q^*$ is a multiplicative function, $(\det M)^n = \det M^n$ and thus any elements of $G$ is equal to
	\begin{equation*}
	X^n = \left(\begin{matrix}
	M^n & 0\\
	0 & \det M^n 
	\end{matrix}\right)\;,
	\end{equation*}
	for some positive integer $n$.
	
	Consider a subspace $U = \langle (1,0,0), (0,1,0) \rangle_\F  \subset V$. Define an $\F$-linear map $f: U \rightarrow V$ as $f((1,0,0)) = (0,1,0)$ and $f((0,1,0)) = (1,0,0)$. As in the proof of \Cref{thm:geq4}, $f(U) = U$ and $U$ intersects only two orbits of $V/G$, $\{0\}$ and $U \setminus \{0\}$.
	
	By contradiction, assume that there exists $g \in \bar{G}$ such that $g = f$ on $U$. Then $g$ has the following form,
	\begin{equation*}
	g = \left(\begin{matrix}
	0 & 1 & 0\\
	1 & 0 & 0\\
	\alpha & \beta & \gamma 
	\end{matrix} \right)\;,
	\end{equation*}
	for some $\alpha, \beta \in \F$, $\gamma \in \mathbb{F}_q^*$.
	Use the fact that $g$ preserves the orbits of $G$. There exists an integer $p_1$ such that $(0,0,1) g = (0,0,1)X^{p_1}$, or the same, $(\alpha, \beta, \gamma) = (0, 0, (\det M)^{p_1})$. Hence $\alpha = \beta = 0$.
	
	For $v = (x, y, z) \in V$ there exists an integer $p_v$ such that $(x, y, z) g = (y, x, \gamma z) = (x, y, z) X^{p_v} = ( (x, y) M^{p_v}, z (\det M)^{p_v})$.
	Put $v = v_1 = (1,1,1)$. Then 
	$$((1,1),\gamma) = ((1, 1)M^{p_{v_1}}, \det M^{p_{v_1}})\;.$$
	Since $M^{p_{v_1}}$ is a multiplication matrix, it fixes a nonzero element in $\mathbb{F}_q^2$ if and only if $M^{p_{v_1}} = I$. From the equality on the third coordinate, $\gamma = \det I = 1$.
	
	Put $v = v_2 = (1,0,1)$. Then we have
	$$((0,1), 1) = ((1, 0)M^{p_{v_2}}, \det M^{p_{v_2}})\;.$$
	We chose $M$ to be such that $(1,0)M = (0,1)$ and therefore $M^{p_{v_2}} = M$. Hence $\det M^{p_{v_2}} = \det M = 1$. Recall $\omega$ is a generator of $\mathbb{F}_{q^2}^*$. It is a well-known fact that the norm $N: \mathbb{F}_{q^2}^* \rightarrow \mathbb{F}_q^*$, defined as $N(\omega^k) = (\det M)^k$ is an onto map, see \cite[pp.~284 -- 291]{lang}.
	Since $\F \neq \mathbb{F}_2$, $\mathbb{F}_q^* \neq \{1\}$, we get a contradiction. Therefore, $V$ is not $G$-pseudo-injective.
\end{proof}

Now, observe the case $q = 2$ and $V = \mathbb{F}_2^3$.
From \Cref{lemma:a2} we know that if $\dim_{\mathbb{F}_2} U \leq 1$, then $f: U \rightarrow V$ that preserves the orbits of $G$ extends to an element of $\bar{G}$ for any subgroup $G$. The same for $U = V$. Therefore we have to check the case $\dim_{\mathbb{F}_2} U = 2$.

Let $G$ be a subgroup of $\GL_3(\mathbb{F}_2)$. Let $U$ be a subspace of dimension $2$, $U = \langle a, b \rangle_{\mathbb{F}_2} = \{ 0, a, b, a+ b \}$. Let $f: U \rightarrow V$ be an injective $\mathbb{F}_2$-linear map preserving the orbits of $G$. Denote $f(a) = c$ and $f(b) = d$. Then $f(a + b) = c+ d$. The map $f$ preserves the orbits of $G$ if and only if the elements in the pairs $(a, c)$, $(b, d)$ and $(a+b, c+ d)$ belong to the same orbits.

Assume that $f$ is unextendable to an element of $\bar{G}$. Since $a$ and $c$ are in the same orbit of $G$, there exists an element $g \in G$ such that $g(c) = a$. Then, the map $gf$ does not extends to an element of $\bar{G}$. From now suppose that $f(a) = a$.

Let $H \in \GL_3(\mathbb{F}_2)$ be a matrix of change of basis such that $aH = (1,0,0)$ and $bH = (0,1,0)$. Then the map $Hf: HU \rightarrow V$ preserves the orbits of the group $G^H = \{ H^{-1}gH \mid g \in G \}$ and is unextendable to an element of $\bar{G^H}$. From now we suppose that $a = (1,0,0)$ and $b = (0,1,0)$.

Let $G_a \leq G$ be a subgroup that contains all the elements of $G$ that fix $a$. From \Cref{thm:closed-subgroup}, $G_a$ is closed and $f$ does not extend to an element of $G_a$. Therefore we can assume that $G$ is a closed group that fixes $a$.
The final computational problem that we have to solve is the following.

\begin{enumerate}
	\item Find all closed subgroups $G$ of $\GL_3(\mathbb{F}_2)$ that have a one-element orbit $O_a = \{(1,0,0)\}$.
	\item For all such $G$, let $O_b \subseteq \mathbb{F}_2^3$ be an orbit of $G$ that contains $(0,1,0)$. Check that for any $c \in O_b$, such that $(1,0,0) + c$ and $(1,1,0)$ are in the same orbit of $G$, there exists $g \in G$ such that $(1,0,0)g = (1,0,0)$ and $(0,1,0)g = c$.
\end{enumerate}

To compute all the closed subgroups of $\GL_3(\mathbb{F}_2)$ use the fact that the closed subgroups are in relation with the partition of $\mathbb{F}_2^3$ into orbits of subgroups in $\GL_3(\mathbb{F}_2)$, see \Cref{thm:poset-closed}. The idea is the following, we first compute the orbits of all cyclic subgroup, then we calculate all the elements in the poset $\mathcal{P}$ built on $\GL_3(\mathbb{F}_2)$. Now, having all possible partitions into orbits we can easily calculate all closed subgroups of $G$.

\begin{proposition}\label{thm:f23}
	For any subgroup $G \leq \GL_3(\mathbb{F}_2)$ the space $\mathbb{F}_2^3$ is $G$-pseudo-injective.
\end{proposition}
\begin{proof}
	The result is computed. See \cite{git-code} for a SAGE source code.
\end{proof}

\subsection{Main result}
\begin{proposition}\label{thm:pi-of-vector-spaces}
	Let $\F$ be a finite field and let $V$ be an $n$-dimensional $\F$-linear vector space. The space $V$ is $G$-pseudo-injective for any subgroup $G \leq \GL(V)$ if and only if $n < 3$ or $V = \mathbb{F}_2^3$.
\end{proposition}
\begin{proof}
	See \Cref{thm:lt2}, \Cref{thm:geq4}, \Cref{thm:dim3} and \Cref{thm:f23}.
\end{proof}

Applying this result to the codes over vector space alphabets we get the following.
\begin{corollary}\label{thm:EP-for-1-length-codes}
	Let $\F$ be a finite field and let $A$ be an $\ell$-dimensional $\F$-linear vector space. There exist a code $C \subset A^1$, a subgroup $G < \GL(A)$ and an unextendable $\swc_G$-isometry $f \in \Hom_\F(C,A^1)$ if and only if $\ell \geq 3$ or $\ell = 3$, $q \neq 2$.
\end{corollary}
\begin{proof}
	The result follows directly from \Cref{thm:pi-ep} and \Cref{thm:pi-of-vector-spaces}.
\end{proof}

\bibliographystyle{acm}
\bibliography{biblio}
\end{document}